\documentclass[reqno,11pt]{amsart}
\usepackage[utf8]{inputenc}
\usepackage[english]{babel}
\usepackage[left=2.5cm,right=2.5cm,top=2cm,bottom=2cm]{geometry}
\usepackage{setspace}

\usepackage{amsmath,amssymb,amsthm,mathtools}
\IfFileExists{stmaryrd.sty}{\usepackage{stmaryrd}}{%
}
\IfFileExists{bbm.sty}{\usepackage{bbm}}{%
} 
\usepackage{graphicx}
\IfFileExists{pstricks.sty}{\usepackage{pstricks,pst-plot}}{}
\usepackage{tikz}
\usetikzlibrary{decorations.markings}
\usepackage{xcolor}
\usepackage{enumerate}
\usepackage{array}
\usepackage[colorlinks=true, linkcolor=blue, citecolor=green, urlcolor=purple]{hyperref}
\usepackage[nameinlink,capitalise]{cleveref}

\usepackage{cancel}
\usepackage{url}
\usepackage[all]{xy}  
\usepackage[pagewise]{lineno}

\def\today{\ifcase\month\or
  January\or February\or March\or April\or May\or June\or
  July\or August\or September\or October\or November\or December\fi
  \space\number\day, \number\year}

\newtheorem{theorem}{Theorem}[section]

\newtheorem{proposition}[theorem]{Proposition}
\newtheorem{corollary}[theorem]{Corollary}
\theoremstyle{definition}
\newtheorem{definition}[theorem]{Definition}
\newtheorem{example}[theorem]{Example}
\theoremstyle{remark}
\newtheorem{remark}[theorem]{Remark}

\newtheorem*{disclaimer}{\textbf{Disclaimer}}

\tikzset{every picture/.style={line width=0.75pt}}  

\begin{document}

\title[About two results for new valid juggling sequences]{About two results for new valid juggling sequences}
 \author[]{Hugo Parada}
\date{\today}

\thanks{Hugo Parada is funded by the Agence Nationale de la Recherche by the QuBiCCS project (ANR-24-CE40-3008)}
 \address{Universit\'e de Lorraine, CNRS, Inria, IECL, F-54000 Nancy, France.}
\email{hugo.parada@inria.fr}
\allowdisplaybreaks
\numberwithin{equation}{section}

\begin{abstract}
In this note, we study two simple operations for extending juggling (siteswap) sequences, called \emph{first throws} and \emph{last catches}. Both constructions come from a natural question for a juggler: how can one add throws at the beginning or catches at the end of a pattern without creating a collision? Using landing times and the permutation test, we give necessary and sufficient conditions under which these constructions, when applied to an arbitrary valid siteswap, produce another valid siteswap. We complement our analysis with several examples. An interactive visualization of these extensions is available in~\cite{parada2026siteswapvisualizer}.
\end{abstract}

\maketitle
%\setcounter{tocdepth}{2}
%\tableofcontents

\section{Introduction and a crash course on siteswap notation}

\noindent
The \emph{siteswap notation} is a mathematical system developed in the 1980s to encode juggling patterns using  finite sequences of nonnegative integers. At its core, it provides a model to describe the timing and trajectory of throws in a periodic juggling \textit{
trick}. Beyond its recreational origins, the siteswap formalism has inspired a rich interplay between combinatorics and modular arithmetic. Siteswap notation serves as a language to communicate juggling sequences between jugglers, but clearly has some limitations; we cannot model all the richness of  juggling using mathematics and, of course, it is not necessary; still, siteswap may help in the construction of  new \textit{tricks}.\bigskip

\noindent  The precise origin of siteswap notation is somewhat diffuse. As Polster points out~\cite[p.~5]{polster2003mathematics}, slightly different accounts of its invention coexist. Nevertheless, the idea of encoding juggling patterns by integer sequences appears to have been developed independently, around the same time, by at least three groups~\cite[p.~507]{buhler1994juggling}:
\begin{itemize}
\item Paul Klimek in Santa Cruz, California, whose work seems to have preceded the other two developments~\cite{tiemann1991notation}.
\item Bruce Tiemann and Bengt Magnusson at the California Institute of Technology~\cite{magnusson1989physics,tiemann1991notation}.
\item Adam Chalcraft, Mike Day, and Colin Wright in Cambridge~\cite{wrightday1995origins}.
\end{itemize}

Now, we introduce siteswap notation, here we follow  \cite{polster2003mathematics}.In siteswap notation, each integer represents the number of beats after which a thrown object will land. 
For instance, the sequence \texttt{441} means: The first throw lands in 4 beats; The second throw lands in 4 beats;  The third throw lands in 1 beat. For an illustration of the temporal structure of the pattern 441, see Figure \ref{fig:beats}.

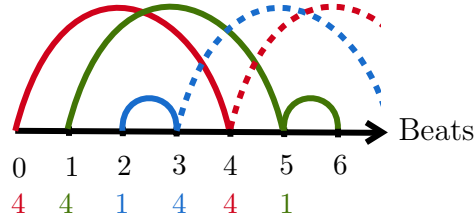
\begin{figure}[h!]
\centering

\begin{tikzpicture}
[x=0.75pt,y=0.75pt,yscale=-1,xscale=1,scale=0.6]

% Beat line
\draw [line width=2.25]
  (126.31,165.96) -- (431.6,167.64)
  (171.33,162.21) -- (171.29,170.21)
  (216.33,162.46) -- (216.29,170.46)
  (261.33,162.71) -- (261.29,170.71)
  (306.33,162.95) -- (306.29,170.95)
  (351.33,163.20) -- (351.29,171.20)
  (396.33,163.45) -- (396.29,171.45);

\draw [
  color={rgb,255:red,208;green,2;blue,27},
  line width=2.25
]
  (126.31,165.96)
  .. controls (166.51,6.52) and (280.12,53.51)
  .. (305.04,164.29);

\draw [
  color={rgb,255:red,65;green,117;blue,5},
  line width=2.25
]
  (170.16,164.29)
  .. controls (210.36,4.84) and (324.97,53.51)
  .. (349.88,164.29);

\draw [
  color={rgb,255:red,24;green,108;blue,204},
  line width=2.25
]
  (216.01,165.96)
  .. controls (214.35,129.04) and (262.18,130.72)
  .. (261.19,164.29);

\draw [
  color={rgb,255:red,65;green,117;blue,5},
  line width=2.25
]
  (351.01,165.96)
  .. controls (349.35,129.04) and (397.18,130.72)
  .. (396.19,164.29);

\begin{scope}
\clip (110,0) rectangle (432,180);

\draw [
  color={rgb,255:red,24;green,108;blue,204},
  line width=2.25,
  dashed
]
  (261.31,165.96)
  .. controls (301.51,6.52) and (415.12,53.51)
  .. (440.04,164.29);

\draw [
  color={rgb,255:red,208;green,2;blue,27},
  line width=2.25,
  dashed
]
  (306.16,164.29)
  .. controls (346.36,4.84) and (459.97,53.51)
  .. (484.88,164.29);
\end{scope}

\draw [line width=2.25]
  (420,159) -- (433.33,167.32) -- (420,175.63);

% Beat labels
\draw (121.31,186.56) node
  [anchor=north west,inner sep=0.75pt] {$0$};

\draw (165.36,185.55) node
  [anchor=north west,inner sep=0.75pt] {$1$};

\draw (209.01,184.88) node
  [anchor=north west,inner sep=0.75pt] {$2$};

\draw (255.05,183.88) node
  [anchor=north west,inner sep=0.75pt] {$3$};

\draw (297.71,184.88) node
  [anchor=north west,inner sep=0.75pt] {$4$};

\draw (344.75,183.88) node
  [anchor=north west,inner sep=0.75pt] {$5$};

\draw (389.75,183.88) node
  [anchor=north west,inner sep=0.75pt] {$6$};

% First 441

% Beat 0: red object, throw 4
\draw (120.31,216.25) node [
  anchor=north west,
  inner sep=0.75pt,
  color={rgb,255:red,208;green,2;blue,27}
] {$4$};

% Beat 1: green object, throw 4
\draw (160.17,216.25) node [
  anchor=north west,
  inner sep=0.75pt,
  color={rgb,255:red,65;green,117;blue,5}
] {$4$};

% Beat 2: blue object, throw 1
\draw (207.01,216.25) node [
  anchor=north west,
  inner sep=0.75pt,
  color={rgb,255:red,24;green,108;blue,204}
] {$1$};

% Second 441

% Beat 3: blue object, throw 4
\draw (255.05,216.25) node [
  anchor=north west,
  inner sep=0.75pt,
  color={rgb,255:red,24;green,108;blue,204}
] {$4$};

% Beat 4: red object, throw 4
\draw (297.71,216.25) node [
  anchor=north west,
  inner sep=0.75pt,
  color={rgb,255:red,208;green,2;blue,27}
] {$4$};

% Beat 5: green object, throw 1
\draw (344.75,216.25) node [
  anchor=north west,
  inner sep=0.75pt,
  color={rgb,255:red,65;green,117;blue,5}
] {$1$};

% Axis label
\draw (438,164.7) node
  [anchor=west,font=\large] {Beats};

\end{tikzpicture}

\caption{Juggling diagram for the siteswap sequence \texttt{441}.}
\label{fig:beats}
\end{figure}

The sequence is then repeated periodically; the notation captures the \emph{temporal rhythm} of a juggling pattern and produces a precise alternation of catches and releases. 

\subsection{Visualizing and validating juggling patterns}

To interpret siteswap sequences, it is helpful to assign a visual meaning to the integers that describe each throw. From now on, we shall use the following convention:
\begin{itemize}
    \item \textbf{0}: No throw (an empty hand).
    \item \textbf{1}: A quick horizontal transfer between hands.
    \item \textbf{2}: A hold, where the object is not thrown.
    \item \textbf{3}: A normal cross throw (to the opposite hand).
    \item \textbf{4}: A high vertical throw (column throw).
    \item \textbf{5}: A high cross throw.
    \item \textbf{Higher numbers:} Throws of increasing height, taking more beats to return; even throws go vertical (to the same hand) and odd throws go cross (to the opposite hand).
    \item For throw heights greater than  \textbf{9}, letters are used: $\texttt{a}=10$, $\texttt{b}=11$, $\texttt{c}=12$, and so on.
\end{itemize}

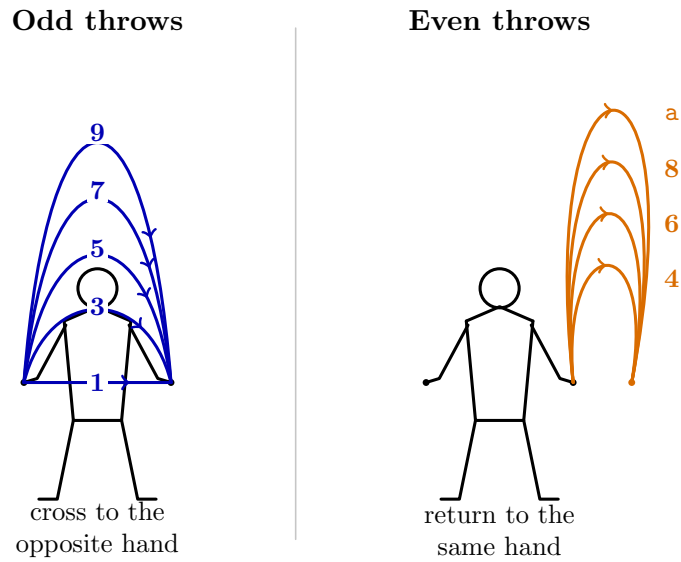
\begin{figure}[ht]
\centering
\begin{tikzpicture}[
  scale=0.76,
  line cap=round,
  line join=round,
  person/.style={
    black,
    line width=1.15pt
  },
  oddthrow/.style={
    blue!70!black,
    line width=1.15pt,
    postaction={decorate},
    decoration={
      markings,
      mark=at position 0.72 with {\arrow{>}}
    }
  },
  eventhrow/.style={
  orange!85!black,
  line width=1.15pt,
  postaction={decorate},
  decoration={
    markings,
    mark=at position 0.50 with {\arrow{>}}
  }
},
  throwlabel/.style={
    font=\bfseries\small,
    fill=white,
    inner sep=1pt
  }
]

% Titles
\node[font=\bfseries] at (0,7.70) {Odd throws};
\node[font=\bfseries] at (7,7.70) {Even throws};

% =================================================
% Odd throws
% =================================================

% Head
\draw[person] (0,3.02) circle (0.34);

% Shoulders
\draw[person]
  (-0.58,2.45) -- (0,2.70) -- (0.58,2.45);

% Torso
\draw[person]
  (-0.58,2.45) --
  (-0.42,0.72) --
  (0.42,0.72) --
  (0.58,2.45);

% Legs
\draw[person]
  (-0.42,0.72) -- (-0.70,-0.65) -- (-1.02,-0.65);

\draw[person]
  (0.42,0.72) -- (0.70,-0.65) -- (1.02,-0.65);

% Arms
\draw[person]
  (-0.55,2.38) -- (-1.05,1.45) -- (-1.28,1.38);

\draw[person]
  (0.55,2.38) -- (1.05,1.45) -- (1.28,1.38);

% Hands
\fill[black] (-1.28,1.38) circle (1.7pt);
\fill[black] (1.28,1.38) circle (1.7pt);

% Odd throw 1
\draw[oddthrow]
  (-1.28,1.38) -- (1.28,1.38);

\node[throwlabel,text=blue!70!black]
  at (0,1.38) {1};

% Odd throw 3
\draw[oddthrow]
  (-1.28,1.38)
  .. controls (-0.85,3.08) and (0.85,3.08) ..
  (1.28,1.38);

\node[throwlabel,text=blue!70!black]
  at (0,2.70) {3};

% Odd throw 5
\draw[oddthrow]
  (-1.28,1.38)
  .. controls (-0.85,4.35) and (0.85,4.35) ..
  (1.28,1.38);

\node[throwlabel,text=blue!70!black]
  at (0,3.72) {5};

% Odd throw 7
\draw[oddthrow]
  (-1.28,1.38)
  .. controls (-0.85,5.65) and (0.85,5.65) ..
  (1.28,1.38);

\node[throwlabel,text=blue!70!black]
  at (0,4.72) {7};

% Odd throw 9
\draw[oddthrow]
  (-1.28,1.38)
  .. controls (-0.85,6.95) and (0.85,6.95) ..
  (1.28,1.38);

\node[throwlabel,text=blue!70!black]
  at (0,5.73) {9};

\node[align=center,font=\small] at (0,-1.18)
  {cross to the\\opposite hand};

%================Even =================================
\begin{scope}[xshift=7cm]

% Head
\draw[person] (0,3.02) circle (0.34);

% Shoulders
\draw[person]
  (-0.58,2.45) -- (0,2.70) -- (0.58,2.45);

% Torso
\draw[person]
  (-0.58,2.45) --
  (-0.42,0.72) --
  (0.42,0.72) --
  (0.58,2.45);

% Legs
\draw[person]
  (-0.42,0.72) -- (-0.70,-0.65) -- (-1.02,-0.65);

\draw[person]
  (0.42,0.72) -- (0.70,-0.65) -- (1.02,-0.65);

% Arms
\draw[person]
  (-0.55,2.38) -- (-1.05,1.45) -- (-1.28,1.38);

\draw[person]
  (0.55,2.38) -- (1.05,1.45) -- (1.28,1.38);

% Hands
\fill[black] (-1.28,1.38) circle (1.7pt);
\fill[black] (1.28,1.38) circle (1.7pt);

% Even throw 4
\draw[eventhrow]
  (1.28,1.38)
  .. controls (1.05,4.10) and (2.75,4.10) ..
  (2.30,1.38);

\node[throwlabel,text=orange!85!black]
  at (3.00,3.20) {4};

  % Even throw 6
\draw[eventhrow]
  (1.28,1.38)
  .. controls (0.90,5.30) and (3.00,5.30) ..
  (2.30,1.38);

\node[throwlabel,text=orange!85!black]
  at (3.00,4.15) {6};

% Even throw 8
\draw[eventhrow]
  (1.28,1.38)
  .. controls (0.75,6.50) and (3.25,6.50) ..
  (2.30,1.38);

\node[throwlabel,text=orange!85!black]
  at (3.00,5.10) {8};

% Even throw 10
\draw[eventhrow]
  (1.28,1.38)
  .. controls (0.60,7.70) and (3.50,7.70) ..
  (2.30,1.38);

\node[throwlabel,text=orange!85!black]
  at (3.00,6.05) {\texttt{a}};

% Landing point
\fill[orange!85!black] (2.30,1.38) circle (1.7pt);
\end{scope}

\node[align=center,font=\small] at (7,-1.18)
  {return to the\\same hand};

% Divider
\draw[gray!45,line width=0.6pt]
  (3.45,-1.35) -- (3.45,7.55);

\end{tikzpicture}

\caption{Visualizing throws.}
\label{fig:throw-parity}
\end{figure}

Each number encodes not only the throw height, but also implicitly determines which hand performs the next catch and at what future beat it will occur. This convention reflects what the audience perceives during an actual juggling performance: a rhythmic alternation of throws of varying heights, each corresponding to a distinct visual trajectory in space and time.

\begin{disclaimer}
Throughout this work and motivated by my juggling background, we use the words \textit{objects} and \textit{balls} as synonyms, with the idea of always visualizing juggling patterns with balls.\footnote{I apologize to jugglers of clubs, rings, hula hoops, and other props.}
\end{disclaimer}

\medskip
\begin{example}Some well-known periodic juggling sequences are:
\begin{itemize}
\item \texttt{31}, \texttt{441}, \texttt{63141}, \texttt{633} (3-balls).
\item \texttt{97531}, \texttt{91} (5-balls), \texttt{756} (6-balls), \texttt{7} (7-balls).
\end{itemize}
\end{example}
\noindent A natural question regarding this notation is: \textit{when is a sequence valid?} 

\begin{definition}[Valid siteswap]
A finite sequence of nonnegative integers represents a valid (simple) juggling pattern if the following conditions are satisfied:
\begin{itemize}
    \item[(J1)] Throws occur at regular, evenly spaced beats.
    \item[(J2)] The sequence is periodic and extends infinitely in both time directions.
    \item[(J3)] At most one ball is thrown, and at most one ball is caught, at each beat.
\end{itemize}   
\end{definition}
If a sequence satisfies (J1) and (J2), it defines a periodic rhythm. When condition (J3) also holds, the sequence is called a \emph{simple juggling pattern}. The fact that we talk about \textit{simple} juggling means that condition (J3) can actually  be violated in the mathematical framework and in real-life juggling! When we allow catching and/or throwing more than one ball in a beat, we talk about \textit{multiplex} or \textit{synchronic} siteswap.

%%%%%%%%%%%%%%%%%%%%%%%%%%%%%%%%%%%
\subsection{Algebraic formulation of the siteswap model}

In this formalism, a juggling pattern is represented by a finite sequence
\[s = (a_0, a_1, \dots, a_{p-1}) \in \mathbb{N}^p,\]
where each $ a_k $ denotes the number of beats after which the ball thrown at time $ k $ will land.  
The sequence is assumed to be periodic with period $ p $, and time is modeled cyclically using the ring $ \mathbb{Z}_p $.

\medskip
A next natural question is: how many balls are required to juggle a given pattern? The answer is given by the following result:

\begin{theorem}[Average Theorem]
Let $ s = (a_0, \dots, a_{p-1}) \in \mathbb{N}^p $ be a valid siteswap sequence. Then the number of balls being juggled equals the average value of the sequence:
\[\text{Number of balls} = \frac{1}{p} \sum_{k=0}^{p-1} a_k.\]
\end{theorem}

\begin{example}\leavevmode
\begin{itemize}
    \item $ \texttt{441} $: $ \frac{4+4+1}{3} = 3 \Rightarrow $ a possible 3-ball pattern.
    \item $ \texttt{5432} $: $ \frac{5+4+3+2}{4} = 3.5 \Rightarrow $ invalid pattern (fractional number of balls).
\end{itemize}
\end{example}

\noindent This provides a powerful necessary condition for validity: if the average is not an integer, the sequence cannot represent a physical juggling pattern\footnote{From the point  of view of objects,  half a ball is still one object=ball.}.

\medskip
A \emph{site swap} is a local operation that exchanges the landing times of two throws, yielding a new valid sequence. 

\begin{definition}(Site swap $s_{i,j}$.)
Let $ s = (a_k)_{k=0}^{p-1} $ be a sequence of nonnegative integers. Given indices $ 0 \le i < j < p $ with $ j - i \le a_i $, we define the site swap $ s_{i,j} $ by:
\[a_i' = a_j + (j - i), \quad a_j' = a_i - (j - i),\]
with all other values unchanged.    
\end{definition}

This new sequence satisfies the following properties. 
\begin{itemize}
    \item[(S1)] $s$ is a valid siteswap if and  only if $s_{i,j}$ is a valid siteswap.
    \item[(S2)] The average value is preserved: the number of balls remains the same.
    \item[(S3)] Site swaps can be used to transform one pattern into another within the same ball count.
\end{itemize}

\begin{example}
\[\texttt{642} \xrightarrow{\text{swap 0 and 1}} \texttt{552}, \qquad
\texttt{642} \xrightarrow{\text{swap 0 and 2}} \texttt{444}\]
\end{example}

\medskip
\begin{definition}(Cyclic shifts.)
Given a sequence $ s = (a_0, \dots, a_{p-1}) $, its \emph{cyclic right shift} is defined as:
\[
s^\rightarrow = (a_{p-1}, a_0, \dots, a_{p-2}).
\]
\end{definition}
This operation preserves validity and ball count:
\begin{itemize}
    \item[(C1)] $ s $ is valid $ \iff $ $ s^\rightarrow $ is valid.
    \item[(C2)] The average is invariant.
    \item[(C3)] The number of balls is unchanged.
\end{itemize}

\medskip
\subsection{The flattening algorithm.}
This is a finite procedure to transform any valid siteswap into a constant sequence (i.e., all entries equal to the average ball count), using only site swaps and cyclic shifts.

\begin{enumerate}
    \item If the sequence is constant, stop.
    \item Otherwise, shift so that the maximal entry $ e $ is at position 0 and a smaller value $ f $ at position 1.
    \item If $e-f=1$, the two throws collide, since $0+e=1+f$; this case cannot occur for a valid siteswap. Otherwise, apply a site swap at positions $0$ and $1$ and repeat.
\end{enumerate}

\begin{example}
\[
\texttt{642} \xrightarrow{\text{swap}} \texttt{552} \xrightarrow{\text{shift}} \texttt{525}
\xrightarrow{\text{swap}} \texttt{345} \xrightarrow{\text{shift}} \texttt{534}
\xrightarrow{\text{swap}} \texttt{444}
\]

\[
\texttt{514} \xrightarrow{\text{swap}} \texttt{244} \xrightarrow{\text{shift}} \texttt{424}
\xrightarrow{\text{swap}} \texttt{334} \xrightarrow{\text{shift}} \texttt{433}
\Rightarrow \text{Collision!}
\]
\end{example}
Any $ b $-ball juggling sequence of period $ p $ can be transformed into the constant sequence $ (b, \dots, b) $ through a finite sequence of site swaps and cyclic shifts.
%%%%%%%%%%%%%%%%%%%%%%%%%%%%%%%%%%%%%%%%%%%%%%
\subsection{The permutation test and vertical shifts}
A fundamental physical constraint is that each throw must land at a distinct future time: otherwise, two balls would need to be caught simultaneously, which is forbidden in simple juggling. This principle translates into a key algebraic characterization.
\begin{figure}[ht]
\centering
\begin{tikzpicture}[x=1.35cm,y=1.1cm]
  \foreach \x in {0,...,3}{
    \node[draw,circle,inner sep=1.5pt] (t\x) at (\x,1) {$\x$};
    \node[draw,circle,inner sep=1.5pt] (l\x) at (\x,0) {$\x$};
  }
  \draw[->,red] (t0)--(l2);
  \draw[->,blue] (t1)--(l1);
  \draw[->,green!60!black] (t2)--(l3);
  \draw[->,purple] (t3)--(l0);
  \node[left] at (-0.35,1) {throws};
  \node[left] at (-0.35,0) {landings};
\end{tikzpicture}
\caption{The landing permutation of \texttt{6451}: $0\mapsto2$, $1\mapsto1$, $2\mapsto3$, $3\mapsto0$.}
\label{fig:permutation-test}
\end{figure}
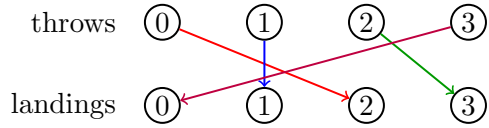
\begin{theorem}[Permutation Test]
Let $ s = (a_0, a_1, \dots, a_{p-1}) \in \mathbb{N}^p $.  
Define the map
\[
\varphi_s : \mathbb{Z}_p \to \mathbb{Z}_p, \qquad \varphi_s(i) = (i + a_i) \bmod p.
\]
Then the sequence $ s $ defines a valid juggling pattern if and only if $ \varphi_s $ is a permutation of $ \mathbb{Z}_p $.
\end{theorem}

\begin{example}
\leavevmode
\begin{itemize}
    \item $ \texttt{6451} $: $ (0+6, 1+4, 2+5, 3+1) \bmod 4 = (2,1,3,0) \Rightarrow $ valid.
    \item $ \texttt{6145} $: $ (0+6, 1+1, 2+4, 3+5) \bmod 4 = (2,2,2,0) \Rightarrow $ not valid (collisions).
\end{itemize}
\end{example}

This result provides a complete algebraic criterion: validity is equivalent to injectivity of the landing time function modulo $ p $.

\medskip
Given a valid sequence $s=(a_k)_{k=0}^{p-1}$ of $b$-balls (average $b$), we may apply a \emph{vertical shift} to obtain a new pattern with more balls.

\begin{definition}[Vertical shift]
Let $ d \in \mathbb{Z} $ be such that
\[d \geq -\min\{a_0, a_1, \dots, a_{p-1}\},\]
and define the new sequence $ s' = (a_k + d)_{k=0}^{p-1} $. Then $s'$ is a $(b+d)$-ball pattern.
\end{definition}

\noindent This operation shifts all throws vertically by the same amount. Since the landing times $ i + a_i $ are increased uniformly, their distinctness modulo $ p $ is preserved.
\medskip

\begin{remark}
Vertical shifts are particularly useful in real-life juggling. They allow performers to convert low, fast patterns into higher ones, without altering the rhythm structure.   
\end{remark}

\begin{example}
\leavevmode
\begin{itemize}
    \item The 3-ball siteswap $ \texttt{531} $ becomes the 4-ball pattern $ \texttt{642} $ under a vertical shift by 1.
    \item The 3-ball siteswap $ \texttt{63141} $ becomes the 7-ball pattern $ \texttt{a7585} $ under a vertical shift by 4.
    \item The 6-ball siteswap $ \texttt{9555} $ becomes the 1-ball pattern $ \texttt{4000} $ under a vertical shift by $-5$.
\end{itemize}

\end{example}
%%%%%%%%%%%%%%

\subsection{Juggling states: hidden memory and transitions}

Freeze a juggling pattern between two consecutive beats. The throws already made belong to the past, but they have left information for the future: some balls are already scheduled to land after one beat, others after two beats, and so on. These pending landings form the hidden memory of the pattern. They determine which throws can be made next without creating a collision. A \emph{juggling state} is a record of this memory \cite[Section~2.8, pp.~44--45]{polster2003mathematics}. Fixing a maximal height $h$, we read the next $h$ beats from nearest to furthest and write a $1$ whenever a ball is scheduled to land and a $0$
otherwise. Thus, a $b$-ball state of height $h$ is a binary word of length $h$ containing exactly $b$ ones.

\begin{example}
The binary state $ \texttt{01101} $ indicates that balls will land in 2, 3, and 5 beats from now.
\end{example}

At each beat, the state is updated as follows:
\begin{itemize}
    \item Inspect the first bit. If it is $0$, no ball lands and the next throw must be $0$. If it is $1$, the landing ball is caught and can be thrown again.
    \item Remove the first bit, shift the sequence to the left, and append a $0$ at the end.
    \item If a throw of height $j$ is made, the $j$-th position of the new state must be unoccupied; replace its $0$ by a $1$.
\end{itemize}

\begin{figure}[ht]
\centering
\[
\begin{array}{c|cccc}
\text{landing in }5\text{ beats} & 0 & \color{blue}1\color{black} & 0 & 0\\
\text{landing in }4\text{ beats} & 0 & 0 & \color{blue}1\color{black} & 0\\
\text{landing in }3\text{ beats} & \color{red}1\color{black} & 0 &\color{green}1\color{black} & \color{blue}1\color{black}\\
\text{landing in }2\text{ beats} & \color{green}1\color{black} & \color{red}1\color{black} & 0 & \color{green}1\color{black}\\
\text{landing in }1\text{ beat}  & \color{blue}1\color{black} & \color{green}1\color{black} & \color{red}1\color{black} & \color{red}1\color{black}\\
\hline
\text{throw} & 3 & 5 & 3 & 1
\end{array}
\]
\caption{State evolution along the siteswap \texttt{531}.}
\label{fig:state-531}
\end{figure}

A throw of height $j$ is allowed if the corresponding landing site is unoccupied.  
Thus, the current state determines which throws are valid at any moment. The set of all juggling states for $b$ balls and maximum throw height $h$ forms the vertex set of a \emph{state graph}. Each vertex corresponds to a binary word of length $h$ with exactly $b$ ones.

\begin{definition}
Let $ \text{Vert}(b,h) $ denote the number of juggling states for $ b $ balls and height $ h $. Then
\[
\text{Vert}(b,h) = \binom{h}{b} = \frac{h!}{b!(h - b)!}.
\]
\end{definition}

Directed edges connect states that are reachable by making a valid throw, and every edge is labeled by the throw height used in the transition; see \Cref{fig:graph1,fig:graph2}.

\begin{figure}[ht]
    \centering
    \begin{tikzpicture}[
      x=1.45cm,y=1.15cm,>=stealth,
      state/.style={draw,rounded corners,font=\ttfamily\small,inner sep=3pt},
      edge/.style={->,font=\small}
    ]
      \node[state] (A) at (0,1.5) {11100};
      \node[state] (B) at (3,2.7) {11010};
      \node[state] (C) at (3,0.3) {11001};
      \node[state] (D) at (6,1.5) {10110};
      \node[state] (E) at (6,3.2) {10101};
      \node[state] (F) at (6,-0.2) {10011};

      \draw[edge] (A) edge[loop left] node {$3$} (A);
      \draw[edge] (A) edge[bend left=12] node[above] {$4$} (B);
      \draw[edge] (B) edge[bend left=12] node[below] {$2$} (A);
      \draw[edge] (A) edge[bend right=12] node[below] {$5$} (C);
      \draw[edge] (B) edge node[above] {$4$} (D);
      \draw[edge] (C) edge[bend right=10] node[right] {$2$} (B);
      \draw[edge] (C) edge node[below] {$3$} (D);
      \draw[edge] (C) edge node[below] {$5$} (F);
      \draw[edge] (D) edge[bend left=15] node[above] {$1$} (A);
    % The excited cycle 51
      \draw[edge] (B) edge[bend left=12]
      node[above] {$5$} (E);
      \draw[edge] (E) edge[bend left=12] node[above] {$1$} (B);
      
    \end{tikzpicture}
    \caption{State graph for 3 balls and height 5 (partial view).}
    \label{fig:graph1}
\end{figure}
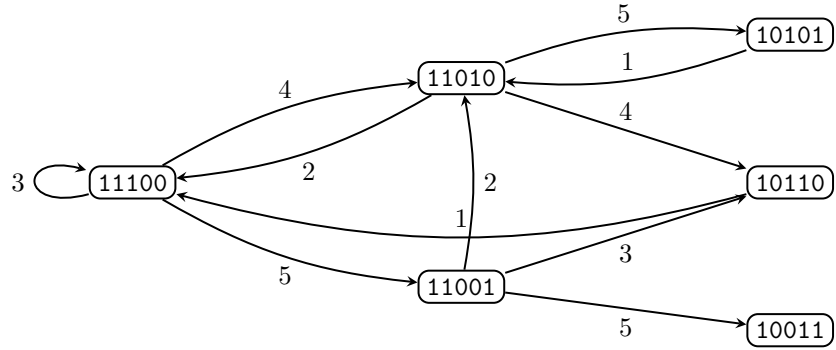

\begin{theorem}
Every closed loop (cycle) in the state graph corresponds to a valid juggling sequence.
\end{theorem}

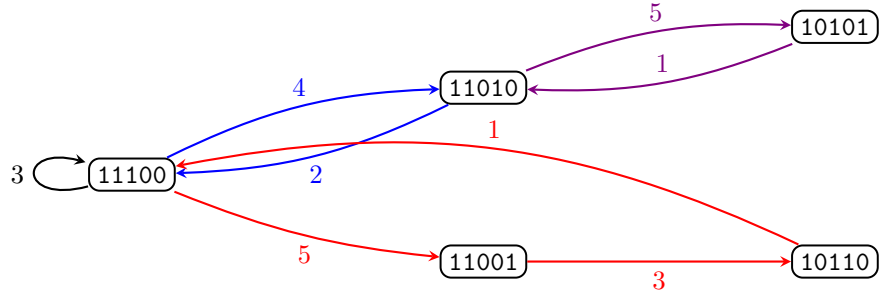
\begin{figure}[ht]
    \centering
    \begin{tikzpicture}[
      x=1.55cm,y=1.15cm,>=stealth,
      state/.style={draw,rounded corners,font=\ttfamily\small,inner sep=3pt},
      edge/.style={->,thick,font=\small}
    ]
      \node[state] (A) at (0,1.5) {11100};
      \node[state] (B) at (3,2.5) {11010};
      \node[state] (C) at (3,0.5) {11001};
      \node[state] (D) at (6,0.5) {10110};
      \node[state] (E) at (6,3.2) {10101};
      \draw[edge,black] (A) edge[loop left] node {$3$} (A);
      \draw[edge,blue] (A) edge[bend left=12] node[above] {$4$} (B);
      \draw[edge,blue] (B) edge[bend left=12] node[below] {$2$} (A);
      \draw[edge,red] (A) edge[bend right=8] node[below] {$5$} (C);
      \draw[edge,red] (C) edge node[below] {$3$} (D);
      \draw[edge,red] (D) edge[bend right=18] node[above] {$1$} (A);
      % The excited cycle 51
      \draw[edge,violet] (B) edge[bend left=12]
      node[above] {$5$} (E);
      \draw[edge,violet] (E) edge[bend left=12] node[above] {$1$} (B);
    \end{tikzpicture}
    \caption{Ground-state cycles \texttt{3}, \texttt{42}, \texttt{531} and excited-state cycle \texttt{51}  in the state graph.}
    \label{fig:graph2}
\end{figure}

We now recall Polster's terminology for ground and excited states and sequences
\cite[Section~2.8.2]{polster2003mathematics}.
\begin{definition}
\leavevmode
\begin{itemize}
\item The ground state is the unique state
$\underbrace{11\cdots1}_{b}
\underbrace{00\cdots0}_{h-b}$, with the $b$ ones placed at the beginning. 
\item Every other
state is called an \emph{excited state}.
\end{itemize}
\end{definition}
The ground state supports the unique
self-loop in the $b$-ball state graph. This loop is labeled by $b$ and corresponds to the constant sequence $\texttt{b}$. For example, when $b=3$
and $h=5$, the ground state is $\texttt{11100}$.

\begin{definition}
A siteswap is called a \emph{ground-state sequence} if its state path originates at the ground state.  Otherwise, it is called an \emph{excited-state sequence}.
\end{definition}

Notice that an excited-state sequence  may or may not visit the ground state. For example, $\texttt{441}$ is a ground-state sequence, whereas $\texttt{414}$ and $\texttt{144}$ are excited-state sequences whose state paths visit the ground state. Thus, after a suitable
cyclic shift, they become ground-state sequences. By contrast, the state path associated with $\texttt{51}$ never visits the ground state. Therefore, to pass from the three-ball cascade to $\texttt{51}$ and return, some transition throws are needed. For instance,
\[\cdots\, \texttt{3333}\, \texttt{4}\, \texttt{51515151}\, \texttt{2}\, \texttt{3333} \,\cdots .\]
Here, the throw $\texttt{4}$ enters the $\texttt{51}$ cycle, whereas
the throw $\texttt{2}$ returns to the ground state.\medskip

We can check easily whether a valid siteswap is based at the ground state.
\begin{proposition}[Ground-state landing times \cite{ChungGraham2008}]\label{prop:ground-state}
Let $s=(a_0,\dots,a_{p-1})$ be a valid $b$-ball siteswap and put
$T_k=k+a_k$. The pattern is based at the ground state if and only if
\[\{T_0,\dots,T_{p-1}\}=\{b,b+1,\dots,b+p-1\}.\]
\end{proposition}
Knowing that some patterns are based at the same state allows us to chain them to create a new pattern. This property is crucial in the construction of long routines.
\begin{example}
\leavevmode
\begin{itemize}
\item All the following patterns start and end at the ground state, and thus can be concatenated freely:
\[ \texttt{3}, \quad \texttt{42}, \quad \texttt{342}, \quad \texttt{441}, \quad \texttt{55500}. \]
\item The pattern $ \texttt{55140} $ decomposes as:
\[
\underbrace{\texttt{551}}_{\text{state } 10110} \quad \texttt{40}.
\]
The sequence $ \texttt{50253} $ starts at the same state $ \texttt{10110} $.  
Hence, one can concatenate:
\[ \texttt{5515025340}.\]
\end{itemize}
\end{example}

\section{Extensions of juggling sequences: first throws and last catches}

In this note, we focus on two algebraic constructions that arise naturally in the context of extending juggling sequences:
\begin{itemize}
    \item[(i)] \textbf{Forward extensions (``first throws''):}
    the first throws of a valid sequence are elevated, and new throws
    are inserted immediately after them. This delays their landing times
    and creates additional beats before the original routine continues.

    \item[(ii)] \textbf{Backward extensions (``last catches''):}
    the throws having the latest landing times are elevated, and lower
    terminal throws are appended to record the corresponding final catches.
\end{itemize}
While these operations are intuitive from the physical viewpoint of juggling, they admit precise algebraic formulations in terms of modular arithmetic. Our main goal is to rigorously determine under which conditions such extensions preserve the validity of a juggling sequence. To the best
of the author's knowledge, the constructions and the necessary and sufficient conditions presented below have not previously been formulated and proved. The initial heuristic behind these constructions was communicated to the author by Cristobal Bascur at a juggling convention a long time ago.
\begin{remark}
At first glance, one might expect the second construction to concern
\emph{last throws}. However, directly mimicking the forward construction
at the end of a sequence does not, in general, preserve the validity of
the siteswap. This is what leads us to formulate the second operation in
terms of \emph{last catches}.
\end{remark}

\subsection*{Forward extensions}
To gain intuition into our first main result, we begin by observing several families of valid juggling sequences. These correspond to patterns with an increasing number of balls and clearly structured periodic behavior. \footnote{This is the way  originally shown to the author by Cristobal Bascur.}\medskip

\begingroup
\small
\setlength{\tabcolsep}{2.7pt}
\renewcommand{\arraystretch}{1.15}

\begin{minipage}[t]{0.11\textwidth}
\centering
\textbf{1 ball}\par\smallskip

\begin{tabular}{|c|}
\hline
20 \\ \hline
300 \\ \hline
4000 \\ \hline
50000 \\ \hline
\end{tabular}
\end{minipage}%
\hfill
\begin{minipage}[t]{0.19\textwidth}
\centering
\textbf{2 balls}\par\smallskip

\begin{tabular}{|c|c|}
\hline
31 & 330 \\ \hline
411 & 4400 \\ \hline
5111 & 55000 \\ \hline
61111 & 660000 \\ \hline
\end{tabular}
\end{minipage}%
\hfill
\begin{minipage}[t]{0.28\textwidth}
\centering
\textbf{3 balls}\par\smallskip

\begin{tabular}{|c|c|c|}
\hline
42 & 441 & 4440 \\ \hline
522 & 5511 & 55500 \\ \hline
6222 & 66111 & 666000 \\ \hline
72222 & 771111 & 7770000 \\ \hline
\end{tabular}
\end{minipage}%
\hfill
\begin{minipage}[t]{0.39\textwidth}
\centering
\textbf{4 balls}\par\smallskip

\begin{tabular}{|c|c|c|c|}
\hline
53 & 552 & 5551 & 55550 \\ \hline
633 & 6622 & 66611 & 666600 \\ \hline
7333 & 77222 & 777111 & 7777000 \\ \hline
83333 & 882222 & 8881111 & 88880000 \\ \hline
\end{tabular}
\end{minipage}
\endgroup
\medskip

For the $i$-th column, increase the first $i$ entries by $1$ and insert one entry $b-i$ after them; repeat this operation row by row.

\begin{center}
\textbf{5 balls}\par\smallskip

\renewcommand{\arraystretch}{1.15}
\begin{tabular}{|c|c|c|c|c|c|}
\hline
$q\backslash i$&$1$ & $2$ & $3$ & $4$ & $5$ \\
\hline
0&\texttt{\textcolor{blue}{5}\textcolor{black}} & \texttt{\textcolor{blue}{55}\textcolor{black}} & \texttt{\textcolor{blue}{555}\textcolor{black}} & \texttt{\textcolor{blue}{5555}\textcolor{black}} & \texttt{\textcolor{blue}{55555}\textcolor{black}} \\
\hline
1&\texttt{\textcolor{blue}{6}\textcolor{red}{4}}  & \texttt{\textcolor{blue}{66}\textcolor{red}{3}} &
\texttt{\textcolor{blue}{666}\textcolor{red}{2}}  &  \texttt{\textcolor{blue}{6666}\textcolor{red}{1}}
&  \texttt{\textcolor{blue}{66666}\textcolor{red}{0}} \\
\hline
2&\texttt{\textcolor{blue}{7}\textcolor{red}{44}}
&
\texttt{\textcolor{blue}{77}\textcolor{red}{33}}
&
\texttt{\textcolor{blue}{777}\textcolor{red}{22}}
&
\texttt{\textcolor{blue}{7777}\textcolor{red}{11}}
&
\texttt{\textcolor{blue}{77777}\textcolor{red}{00}}
\\
\hline
3&\texttt{\textcolor{blue}{8}\textcolor{red}{444}}
&
\texttt{\textcolor{blue}{88}\textcolor{red}{333}}
&
\texttt{\textcolor{blue}{888}\textcolor{red}{222}}
&
\texttt{\textcolor{blue}{8888}\textcolor{red}{111}}
&
\texttt{\textcolor{blue}{88888}\textcolor{red}{000}}
\\
\hline

4&\texttt{\textcolor{blue}{9}\textcolor{red}{4444}}
&
\texttt{\textcolor{blue}{99}\textcolor{red}{3333}}
&
\texttt{\textcolor{blue}{999}\textcolor{red}{2222}}
&
\texttt{\textcolor{blue}{9999}\textcolor{red}{1111}}
&
\texttt{\textcolor{blue}{99999}\textcolor{red}{0000}}
\\
\hline
\end{tabular}
\end{center}

We now formalize the first-throw construction and state our first main result.
\begin{definition}[First-throw extension]
\label{def:first-throw}
Let $p,b\in\mathbb N$, let $q\geq1$, and let
$s=(a_0,\dots,a_{p-1})\in\mathbb N^p$ be a valid $b$-ball siteswap of period $p$. For $0\leq i\leq\min\{b,p\}$, define
\[F_{i,q}(s)=
(a_0+q,\dots,a_{i-1}+q,
\underbrace{b-i,\dots,b-i}_{q\text{ times}},
a_i,\dots,a_{p-1}).\]
\end{definition}

\begin{theorem}
\label{thm:forward-criterion}
With the notation of Definition~\ref{def:first-throw}, assume that $s$ is a valid $b$-ball siteswap. Let $T_k=k+a_k$, $0\leq k\leq p-1$, be its landing times. For each $r\in\mathbb Z_{p+q}$, define
\[m_r=\#\bigl\{k\in\{0,\dots,p-1\}:T_k\equiv r\pmod{p+q}\bigr\}.\]

Then the following statements are equivalent:
\begin{enumerate}
\item[(i)] $F_{i,q}(s)$ is valid for some
$i\in\{0,\dots,\min\{b,p\}\}$.
\item[(ii)] $F_{i,q}(s)$ is valid for every $i\in\{0,\dots,\min\{b,p\}\}$.
\item[(iii)] The landing times $T_0,\dots,T_{p-1}$ are pairwise distinct modulo $p+q$ and
\[T_k\not\equiv b-q,b-q+1,\dots,b-1\pmod{p+q}, \qquad \text{ for every } k\in\{0,\dots,p-1\}.\]
\item[(iv)] For every $r\in\mathbb Z_{p+q}$,
\[m_r=
\begin{cases}
0,& r\in\{b-q,\dots,b-1\}\pmod{p+q},\\[1mm]
1,& r\notin\{b-q,\dots,b-1\}\pmod{p+q}.
\end{cases}\]
\end{enumerate}

Whenever these conditions hold, $F_{i,q}(s)$ is a $b$-ball siteswap of period $p+q$.
\end{theorem}

\begin{proof}
Fix an admissible $i$. If $k<i$, the corresponding old throw in
$F_{i,q}(s)$ lands at
\[
k+(a_k+q)=T_k+q.
\]
If $k\geq i$, the throw $a_k$ moves from time $k$ to time $k+q$,
and therefore also lands at
\[k+q+a_k=T_k+q.\]
Thus, all the old throws land at $T_0+q,\dots,T_{p-1}+q$. The $q$ new throws are made at times $i,i+1,\dots,i+q-1$, all with height $b-i$. Their landing times are $b,b+1,\dots,b+q-1$. The old landing times are pairwise distinct modulo $p+q$ if and only
if $T_0,\dots,T_{p-1}$ are pairwise distinct modulo $p+q$. Moreover,
an old landing time collides with a new one if and only if
\[T_k+q\equiv b+r\pmod{p+q}\]
for some $r\in\{0,\dots,q-1\}$, or equivalently,
\[T_k\equiv b-q+r\pmod{p+q}.\]
This proves the equivalence of (i), (ii), and (iii). Notice that the condition does not depend on $i$.

There are exactly $p$ residues outside
\[\{b-q,\dots,b-1\}\pmod{p+q}.\]
Hence, condition (iii) holds if and only if each of these residues is
occupied exactly once and each forbidden residue is not occupied.
This proves the equivalence with (iv).  Finally,
\[\begin{aligned}
\sum_{j=0}^{p+q-1}F_{i,q}(s)_j
&=\sum_{k=0}^{p-1}a_k+iq+q(b-i)\\
&=bp+iq+q(b-i)\\
&=b(p+q).
\end{aligned}\]
Therefore, whenever $F_{i,q}(s)$ is valid, it is still a $b$-ball
siteswap.
\end{proof}

\begin{corollary}[Ground-state first throws]
\label{cor:forward-ground-state}
Let $s=(a_0,\dots,a_{p-1})$ be a $b$-ball siteswap based at the ground state. Then, for every $q\geq1$ and every $0\leq i\leq\min\{b,p\}$, the sequence $F_{i,q}(s)$ is a valid $b$-ball siteswap of period
$p+q$, and it is again based at the ground state.
\end{corollary}

\begin{proof}
By Proposition~\ref{prop:ground-state}, $\{T_0,\dots,T_{p-1}\}=\{b,b+1,\dots,b+p-1\}$. The old throws in $F_{i,q}(s)$ land at $\{b+q,b+q+1,\dots,b+p+q-1\}$, whereas the new throws land at
$\{b,b+1,\dots,b+q-1\}$. Together, their landing times are exactly $\{b,b+1,\dots,b+p+q-1\}$. Proposition~\ref{prop:ground-state} gives the result.
\end{proof}

\begin{example}
\leavevmode
\begin{itemize}
\item Consider the excited siteswap $\texttt{51}$. Here $p=2$, $b=3$, $q=1$, and its landing times are $T_0=5$, $T_1=2$. Since $T_0\equiv T_1\equiv2\pmod 3$, the sequence $F_{i,1}(\texttt{51})$ is not valid for any admissible
$i$.
\item Now consider the cyclic shift $\texttt{15}$. Its landing times are $T_0=1$, $T_1=6$, which are congruent to $1$ and $0$ modulo $3$. They are distinct and
avoid the forbidden residue $2$. Therefore,
$F_{i,1}(\texttt{15})$ is valid for every admissible $i$. For example, $F_{0,1}(\texttt{15})=\texttt{315}$.
\end{itemize}
\end{example}
\begin{remark}
The siteswap $\texttt{15}$ is excited, so the ground-state assumption in Corollary~\ref{cor:forward-ground-state} is sufficient but not necessary.
\end{remark}
%%%%%%%%%%%%%%%%%%%%%%%%%%%%%%%%%%
\subsection*{Backward extensions}
As in the forward case, we begin with some families of valid siteswaps to give intuition for the last-catch construction. The first tables show
patterns with one to four balls, while the five-ball table illustrates how the construction acts on the pattern \texttt{66661}.

\noindent
\begingroup
\small
\setlength{\tabcolsep}{2.7pt}
\renewcommand{\arraystretch}{1.15}

\begin{minipage}[t]{0.11\textwidth}
\centering
\textbf{1 ball}\par\smallskip

\begin{tabular}{|c|}
\hline
20 \\ \hline
300 \\ \hline
4000 \\ \hline
50000 \\ \hline
\end{tabular}
\end{minipage}%
\hfill
\begin{minipage}[t]{0.19\textwidth}
\centering
\textbf{2 balls}\par\smallskip

\begin{tabular}{|c|c|}
\hline
31 & 330 \\ \hline
411 & 4400 \\ \hline
5111 & 55000 \\ \hline
61111 & 660000 \\ \hline
\end{tabular}
\end{minipage}%
\hfill
\begin{minipage}[t]{0.28\textwidth}
\centering
\textbf{3 balls}\par\smallskip

\begin{tabular}{|c|c|c|}
\hline
42 & 441 & 4440 \\ \hline
522 & 5511 & 55500 \\ \hline
6222 & 66111 & 666000 \\ \hline
72222 & 771111 & 7770000 \\ \hline
\end{tabular}
\end{minipage}%
\hfill
\begin{minipage}[t]{0.39\textwidth}
\centering
\textbf{4 balls}\par\smallskip

\begin{tabular}{|c|c|c|c|}
\hline
53 & 552 & 5551 & 55550 \\ \hline
633 & 6622 & 66611 & 666600 \\ \hline
7333 & 77222 & 777111 & 7777000 \\ \hline
83333 & 882222 & 8881111 & 88880000 \\ \hline
\end{tabular}
\end{minipage}

\endgroup

\bigskip

\begin{center}
\textbf{5 balls: \texttt{66661}}\par\smallskip

\begingroup
\small
\setlength{\tabcolsep}{3.2pt}
\renewcommand{\arraystretch}{1.18}

\begin{tabular}{|c|c|c|c|c|c|}
\hline
$q\backslash i$ & $1$ & $2$ & $3$ & $4$ & $5$ \\
\hline

$0$
& \texttt{666\textcolor{blue}{6}1}
& \texttt{66\textcolor{blue}{66}1}
& \texttt{6\textcolor{blue}{666}1}
& \texttt{\textcolor{blue}{6666}1}
& \texttt{\textcolor{blue}{66661}}
\\
\hline

$1$
& \texttt{666\textcolor{blue}{7}1\textcolor{red}{4}}
& \texttt{66\textcolor{blue}{77}1\textcolor{red}{3}}
& \texttt{6\textcolor{blue}{777}1\textcolor{red}{2}}
& \texttt{\textcolor{blue}{7777}1\textcolor{red}{1}}
& \texttt{\textcolor{blue}{77772}\textcolor{red}{0}}
\\
\hline

$2$
& \texttt{666\textcolor{blue}{8}1\textcolor{red}{44}}
& \texttt{66\textcolor{blue}{88}1\textcolor{red}{33}}
& \texttt{6\textcolor{blue}{888}1\textcolor{red}{22}}
& \texttt{\textcolor{blue}{8888}1\textcolor{red}{11}}
& \texttt{\textcolor{blue}{88883}\textcolor{red}{00}}
\\
\hline

$3$
& \texttt{666\textcolor{blue}{9}1\textcolor{red}{444}}
& \texttt{66\textcolor{blue}{99}1\textcolor{red}{333}}
& \texttt{6\textcolor{blue}{999}1\textcolor{red}{222}}
& \texttt{\textcolor{blue}{9999}1\textcolor{red}{111}}
& \texttt{\textcolor{blue}{99994}\textcolor{red}{000}}
\\
\hline
\end{tabular}

\endgroup
\end{center}

\smallskip

\noindent
The columns correspond to the number $i$ of selected throws, and the
rows correspond to the value of $q$. The selected throws are shown in
\textcolor{blue}{blue}, the appended throws of height $5-i$ in
\textcolor{red}{red}, and the unchanged throws in black.
\begin{definition}[Last-catch extension]
\label{def:last-catch}
Let $p,b\in\mathbb N$, let $q\geq1$, and let
$s=(a_0,\dots,a_{p-1})$ be a valid $b$-ball siteswap. Put $T_k=k+a_k$, $0\leq k\leq p-1$, and fix $1\leq i\leq\min\{b,p\}$. Let $K_i(s)\subset\{0,\dots,p-1\}$ be the set of indices of the
$i$ largest landing times $T_k$. The \emph{last-catch extension} of $s$ is the sequence
\[B_{i,q}(s)=(\widehat a_0,\dots,\widehat a_{p+q-1}),\]
defined by
\[\widehat a_k=
\begin{cases}
a_k+q, & k\in K_i(s),\\[1mm]
a_k, & k\in\{0,\dots,p-1\}\setminus K_i(s),\\[1mm]
b-i, & k\in\{p,\dots,p+q-1\}.
\end{cases}\]
\end{definition}
Since $s$ is valid, the landing times $T_k$ are pairwise distinct; hence, $K_i(s)$ is well-defined.
\begin{example}
For the pattern $\texttt{67561}$, the two largest landing times
correspond to the throws $\texttt{7}$ and $\texttt{6}$. Hence,
\[\mathtt{6}
\underbrace{\textcolor{blue}{\mathtt{7}}}_{\text{penultimate landing}}
\mathtt{5}
\underbrace{\textcolor{blue}{\mathtt{6}}}_{\text{last landing}}
\mathtt{1}
\;\xrightarrow{i=2,\;q=3}\;
\mathtt{6}
\textcolor{blue}{\mathtt{a}}
\mathtt{5}
\textcolor{blue}{\mathtt{9}}
\mathtt{1}
\textcolor{red}{\mathtt{333}}.\]
\end{example}

\begin{theorem}[Last-catch criterion]
\label{thm:backward-criterion}
With the notation of Definition~\ref{def:last-catch}, let  $K=K_i(s)$ and define the modified landing times of the old throws by
\[\widehat T_k=
\begin{cases}
T_k+q, & k\in K,\\[1mm]
T_k, & k\notin K.
\end{cases}\]
For each $r\in\mathbb Z_{p+q}$, set
\[\widehat m_r=\#\bigl\{k\in\{0,\dots,p-1\}:\widehat T_k\equiv r\pmod{p+q}\bigr\}.\]
Then the following statements are equivalent:
\begin{enumerate}
\item[(i)] The last-catch extension $B_{i,q}(s)$ is valid.
\item[(ii)] The modified landing times $\widehat T_0,\dots,\widehat T_{p-1}$ are pairwise distinct modulo $p+q$ and none of them belongs to
\[\{p+b-i,\dots,p+b-i+q-1\} \pmod{p+q}.\]
\item[(iii)] For every $r\in\mathbb Z_{p+q}$,
\[\widehat m_r=
\begin{cases}
0,& r\in \{p+b-i,\dots,p+b-i+q-1\}\pmod{p+q},\\[1mm]
1,& r\notin \{p+b-i,\dots,p+b-i+q-1\}\pmod{p+q}.
\end{cases}\]
\end{enumerate}
Whenever these conditions hold, $B_{i,q}(s)$ is a $b$-ball siteswap of period $p+q$.
\end{theorem}
\begin{proof}
The landing times of the old throws in $B_{i,q}(s)$ are easily identified. If $k\in K$, the height $a_k$ is increased by $q$, and
the corresponding landing time becomes
\[k+(a_k+q)=T_k+q=\widehat T_k.\]
If $k\notin K$, the throw is unchanged and its landing time remains
\[k+a_k=T_k=\widehat T_k.\]
The $q$ new throws are made at times $p,p+1,\dots,p+q-1$ and all have height $b-i$. Hence, their landing times are
\[p+b-i,\ p+b-i+1,\dots,p+b-i+q-1.\]
Therefore, the complete set of landing times of $B_{i,q}(s)$ is
\[\{\widehat T_0,\dots,\widehat T_{p-1}\}\cup\{p+b-i,\dots,p+b-i+q-1\}.\]
The $q$ landing times in the second set are pairwise distinct modulo
$p+q$. By the permutation test, this union is a complete system of
residues modulo $p+q$ if and only if the modified landing times
$\widehat T_0,\dots,\widehat T_{p-1}$ are pairwise distinct and avoid
the second set. This proves the equivalence between \textup{(i)} and
\textup{(ii)}.

There are exactly $p$ residues outside
\[\{p+b-i,\dots,p+b-i+q-1\}\pmod{p+q}.\]
Thus, condition \textup{(ii)} holds if and only if every residue
outside this set is occupied exactly once by the modified landing
times and every residue inside it is not occupied. This is precisely
condition \textup{(iii)}.

Finally,
\[\begin{aligned}
\sum_{k=0}^{p+q-1}\widehat a_k
&=\sum_{k=0}^{p-1}a_k+iq+q(b-i)\\
&=bp+iq+q(b-i)\\
&=b(p+q).
\end{aligned}\]
Hence, whenever $B_{i,q}(s)$ is valid, it is a $b$-ball siteswap.
\end{proof}
\begin{corollary}[Ground-state last catches]
\label{cor:backward-ground-state}
Let $s=(a_0,\dots,a_{p-1})$ be a $b$-ball siteswap based at the ground state. Then, for every
$q\geq1$ and every $1\leq i\leq\min\{b,p\}$, the last-catch extension $B_{i,q}(s)$ is a valid $b$-ball siteswap
of period $p+q$, and it is again based at the ground state.
\end{corollary}
\begin{proof}
By Proposition~\ref{prop:ground-state}, the landing times of $s$ are $\{T_0,\dots,T_{p-1}\}=
\{b,b+1,\dots,b+p-1\}$. Therefore, the $i$ largest landing times are $\{b+p-i,\dots,b+p-1\}$. The landing times that are not among the $i$ largest remain unchanged
and form the interval
\[\{b,b+1,\dots,b+p-i-1\},\]
with the first interval omitted when $i=p$. The $i$ selected landing times are increased by $q$ and become
\[\{b+p-i+q,\dots,b+p+q-1\}.\]
The $q$ appended throws have height $b-i$ and are made at times $p,\dots,p+q-1$. Hence, their landing times are
\[\{b+p-i,\dots,b+p-i+q-1\}.\]
These three intervals are disjoint, and their union is $\{b,b+1,\dots,b+p+q-1\}$. Proposition~\ref{prop:ground-state} now shows that $B_{i,q}(s)$ is valid and is again based at the ground state. Its number of balls is
$b$ by \Cref{thm:backward-criterion}.
\end{proof}

\begin{example}
\leavevmode
\begin{itemize}
\item Consider the excited siteswap $\texttt{15}$ with $i=q=1$. Its landing times are $T_0=1$, $T_1=6$. The largest landing time is $T_1$, so the corresponding throw is increased by $1$, and a new throw of height $b-i=2$ is appended. Thus,  $B_{1,1}(\texttt{15})=\texttt{162}$, which is not valid.
\item Now consider the cyclic shift $\texttt{51}$. Its landing times are $T_0=5$,  $T_1=2$. The largest landing time is now $T_0$, and
$B_{1,1}(\texttt{51})=\texttt{612}$, which is valid.
\end{itemize}
\end{example}
\begin{remark}
The ground-state assumption in Corollary~\ref{cor:backward-ground-state} is sufficient but not necessary.   
\end{remark}

We also recall that an interactive visualizer for these two extensions is available in~\cite{parada2026siteswapvisualizer}. It is based on the online version of \textit{Juggling Lab}~\cite{boyce2026jugglinglab}. We recommend visiting the site to learn more and to play with siteswaps.
%%%%%%%

\subsection*{More examples}
We now illustrate the two constructions on the same pattern. Let
\[s=(6,4,5,1)\]
with period $p=4$. Its average is  $4$, so $s$ is a $b=4$ ball pattern. Its landing map is $\varphi_s(k)=(k+a_k)\bmod 4$ and
\[\varphi_s=(2,1,3,0).\]
Therefore, $s$ is valid. Moreover, its landing times are $T=(6,5,7,4)$  and hence
\[\{T_0,T_1,T_2,T_3\}=\{4,5,6,7\}.\]
Thus, $s$ is based at the ground state.

\smallskip
\emph{Forward extension.}
Take $i=q=2$. According to
Definition~\ref{def:first-throw}, we increase the first two throws by
$2$ and insert two new throws of height $b-i=2$ after them. Therefore,
\[\begin{aligned}
F_{2,2}(s)
&=
\bigl(
\textcolor{blue}{6+2},
\textcolor{blue}{4+2},
\textcolor{red}{2},
\textcolor{red}{2},
5,1
\bigr)\\
&=
\bigl(
\textcolor{blue}{8},
\textcolor{blue}{6},
\textcolor{red}{2},
\textcolor{red}{2},
5,1
\bigr).
\end{aligned}\]
The old landing times become  $T+2=(8,7,9,6)$, while the new throws land at times $4$ and $5$. Together, the landing
times are
\[\{4,5,6,7,8,9\}.\]
Equivalently, the new landing map is
\[\varphi_{F_{2,2}(s)}=(2,1,4,5,3,0),\]
which is a permutation of $\mathbb Z_6$. Hence, $F_{2,2}(s)$ is valid and is again based at the ground state.

\smallskip
\emph{Backward extension.}
Again take $i=q=2$. The two largest landing times are $T_2=7$ and $T_0=6$, so $K_2(s)=\{0,2\}$. According to Definition~\ref{def:last-catch}, the corresponding throws are increased by $2$, and two new throws of height $b-i=2$ are appended. Thus,
\[\begin{aligned}
B_{2,2}(s)
&=
\bigl(
\textcolor{blue}{6+2},
4,
\textcolor{blue}{5+2},
1,
\textcolor{red}{2},
\textcolor{red}{2}
\bigr)\\
&=
\bigl(
\textcolor{blue}{8},
4,
\textcolor{blue}{7},
1,
\textcolor{red}{2},
\textcolor{red}{2}
\bigr).
\end{aligned}\]
The modified landing times of the old throws are $\widehat T=(8,5,9,4)$ and the appended throws land at times $6$ and $7$. Therefore, the
complete set of landing times is again
\[\{4,5,6,7,8,9\}.\]
Equivalently,
\[\varphi_{B_{2,2}(s)}=(2,5,3,4,0,1),\]
which is a permutation of $\mathbb Z_6$. Hence, $B_{2,2}(s)$ is also valid and based at the ground state.\medskip

We finish the analysis of this work by constructing a sequence/routine that combines forward extensions, state-loop splicing, local landing-time swaps, and a final backward extension. \medskip

\noindent The constant five-ball cascade can be represented by the four-beat
sequence $\mathtt{5555}$. Its first forward extension is
\[F_{4,1}(\mathtt{5555})=\mathtt{66661}.\]
From $\mathtt{66661}$, we apply a last-catch extension with $i=2$ and $q=1$. The two largest landing times correspond to the third and fourth throws, and therefore
\[B_{2,1}(\mathtt{66661})=\mathtt{667713}=\mathtt{66\,771\,3}.\]
The appearance of the block $\mathtt{771}$ is not accidental. In the
five-ball state graph of height $7$, let $G=\mathtt{1111100}$ be the ground state. After the initial throws $\mathtt{66}$, the
juggler reaches the state
\[G\xrightarrow{\mathtt{66}}X, \qquad X=\mathtt{1110110}.\]
At this state, $\mathtt{771}$ is a loop, while both $\mathtt{3}$ and
$\mathtt{661}$ return to the ground state:
\[X\xrightarrow{\mathtt{771}}X,
\qquad X\xrightarrow{\mathtt{3}}G, \qquad X\xrightarrow{\mathtt{661}}G.\]
Indeed, $\mathtt{663}$ and $\mathtt{66661}$ are both ground-state
sequences and share the same initial root $\mathtt{66}$. Hence, after performing the loop $\mathtt{771}$, we may replace the terminal
$\mathtt{3}$ by $\mathtt{661}$:
\[\mathtt{66\,771\,3}\longmapsto \mathtt{66\,771\,661}.\]
We now apply a second forward extension with $i=2$ and $q=3$:
\[F_{2,3}(\mathtt{66\,771\,661})=\mathtt{99333\,771\,661}.\]
Next, for $n\geq1$ we use the local swap rule
\[nm\longmapsto(m+1)(n-1).\]
More precisely, $\mathtt{99}\longmapsto\mathtt{a8}$,
while $\mathtt{333}\longmapsto\mathtt{531}$ as both are ground-state sequences of $3-$balls\footnote{You can justify using the state graph or doing site swap operations.}. Consequently,
\[\mathtt{99333\,771\,661} \longmapsto \mathtt{a8531\,771\,661}.\]
After completing this sequence, the juggler returns to the ground
state. We may therefore append three throws from the basic cascade and
consider
\[\bar s=\mathtt{a8531\,771\,661\,555}.\]
Its landing times are
\[(10,9,7,6,5,12,13,8,14,15,11,16,17,18)\]
and hence form the set $\{5,6,\dots,18\}$. The five largest landing times correspond to the throws
\[\mathtt{6},\quad\mathtt{6},\quad
\mathtt{5},\quad\mathtt{5},\quad\mathtt{5}.\]
Applying the last-catch extension with $i=5$ and $q=2$, these throws
become
\[\mathtt{6},\mathtt{6},\mathtt{5},\mathtt{5},\mathtt{5}
\longmapsto
\mathtt{8},\mathtt{8},\mathtt{7},\mathtt{7},\mathtt{7},\]
and two terminal throws of height $b-i=0$ are appended. Therefore,
\[B_{5,2}(\bar s)=\mathtt{a8531\,771\,88177700}.\]
The complete construction is thus
\[\begin{aligned}
\mathtt{5555}
&\xrightarrow{F_{4,1}}
\mathtt{66661}
\xrightarrow{B_{2,1}}
\mathtt{66\,771\,3}\\
&\longmapsto
\mathtt{66\,771\,661}
\xrightarrow{F_{2,3}}
\mathtt{99333\,771\,661}\\
&\longmapsto
\mathtt{a8531\,771\,661}\\
&\longmapsto
\mathtt{a8531\,771\,661\,555}
\xrightarrow{B_{5,2}}
\mathtt{a8531\,771\,88177700}.
\end{aligned}\]

A video demonstration of this routine is available in \cite{truco5}.

\section{Conclusion}

We introduced two operations for extending valid siteswaps: first throws and last catches. In both cases, we gave necessary and sufficient conditions in terms of the landing times modulo the new period. For patterns based at the ground state, these conditions always hold, and both constructions produce new $b$-ball siteswaps based again at the
ground state. For arbitrary valid siteswaps, the situation is more delicate. Finally, the analysis is  completed  with a more realistic example. A natural next question is whether similar constructions can be developed for synchronous and multiplex siteswaps.

\section*{Acknowledgments}
I wish to thank Cristobal Bascur for introducing me to the worlds of mathematics and juggling. I also thank my juggling friends Fernando Torres, Matias Alarc\'on, Claudio Tapia, and Matias N\'u\~nez for helping me recover my love of juggling.

\nocite{polster2003mathematics,beek1995science,walker1982short,shannon1981scientific}
\bibliographystyle{alpha}
\bibliography{siteswap}

\begin{thebibliography}{BEGW94}

\bibitem[BEGW94]{buhler1994juggling}
Joe Buhler, David Eisenbud, Ron Graham, and Colin Wright.
\newblock Juggling drops and descents.
\newblock {\em The American Mathematical Monthly}, 101(6):507--519, 1994.

\bibitem[BJ26]{boyce2026jugglinglab}
Jack Boyce and {Juggling Lab contributors}.
\newblock Juggling lab.
\newblock Computer software, version 1.7.5, 2026.
\newblock Available at \url{https://jugglinglab.org/}. Accessed July 28, 2026.

\bibitem[BL95]{beek1995science}
Peter~J. Beek and Arthur Lewbel.
\newblock The science of juggling.
\newblock {\em Scientific American}, 273(5):92--97, 1995.

\bibitem[CG08]{ChungGraham2008}
Fan Chung and Ron Graham.
\newblock Primitive juggling sequences.
\newblock {\em The American Mathematical Monthly}, 115(3):185--194, 2008.

\bibitem[MT89]{magnusson1989physics}
Bengt Magnusson and Bruce Tiemann.
\newblock The physics of juggling.
\newblock {\em The Physics Teacher}, 27(8):584--589, 1989.

\bibitem[Para]{parada2026siteswapvisualizer}
Hugo Parada.
\newblock Siteswap extension visualizer.
\newblock
  \url{https://sites.google.com/view/hugo-parada/juggling-outreach/siteswap-extensions}.

\bibitem[Parb]{truco5}
Hugo Parada.
\newblock Siteswap \texttt{a8531 771 88177700}.
\newblock
  \url{https://drive.google.com/file/d/1agN4oUrzV3i3ucQqtWmatr7iOcLwWUE1/view?usp=drive_link}.

\bibitem[Pol03]{polster2003mathematics}
Burkard Polster.
\newblock {\em The Mathematics of Juggling}.
\newblock Springer, New York, 2003.

\bibitem[Sha93]{shannon1981scientific}
Claude~E. Shannon.
\newblock Scientific aspects of juggling.
\newblock In N.~J.~A. Sloane and Aaron~D. Wyner, editors, {\em Claude Elwood
  Shannon: Collected Papers}, pages 850--864. IEEE Press, New York, 1993.
\newblock Manuscript written circa 1980.

\bibitem[TM91]{tiemann1991notation}
Bruce Tiemann and Bengt Magnusson.
\newblock A notation for juggling tricks: A lot of juggling tricks.
\newblock {\em Juggler's World}, 43(2):31--33, 1991.
\newblock Available at
  \url{https://www.dev.juggle.org/history/archives/jugmags/43-2/43-2,p31.htm}.

\bibitem[Wal82]{walker1982short}
Jeff Walker.
\newblock Variations for numbers jugglers.
\newblock {\em Juggler's World}, 34(1):11, 1982.
\newblock Available at
  \url{https://www.jonglage.net/theorie/notation/ladder/refs/Jeff Walker -
  Variations for numbers jugglers - JW-vol34no1-p11.pdf}.

\bibitem[WD95]{wrightday1995origins}
Colin~D. Wright and Mike Day.
\newblock Responses in ``all site swap original inventors sought''.
\newblock Posts to the \texttt{rec.juggling} newsgroup, 1995.
\newblock Messages posted on 23--24 May 1995. Archived by Google Groups at
  \url{https://groups.google.com/g/rec.juggling/c/lB-jy9wcgNs}.

\end{thebibliography}

\end{document}